\documentclass[conference]{IEEEtran}
\usepackage{amsmath,amssymb,amsfonts}
\usepackage{comment}
\usepackage{booktabs}
\usepackage{caption}
\usepackage{algorithm}
\usepackage{algorithmic}
\usepackage{float}
\usepackage{bm}
\usepackage{graphicx}
\usepackage{dblfloatfix}
\usepackage{cite}
\usepackage{placeins}
\usepackage{balance}
\usepackage{tikz}
\usepackage{amsthm}
\usepackage{xcolor}

\usetikzlibrary{shapes.geometric, arrows.meta, positioning, calc, decorations.pathmorphing}
\usepackage[colorlinks=true,linkcolor=blue,citecolor=blue]{hyperref}

\newcommand{\E}{\mathbb{E}}
\newcommand{\Prob}{\mathrm{Pr}}
\newcommand{\Rq}{R_q}
\newcommand{\calH}{\mathcal{H}}

\newtheorem{proposition}{Proposition}
\newtheorem{lemma}{Lemma}
\newtheorem{theorem}{Theorem}
\begin{document}

\title{Privacy-Enhanced Zero-Order Federated Learning via xMK-CKKS over Wireless Channels}

\author{
\IEEEauthorblockN{Anthony Ayli}
\IEEEauthorblockA{anthony.ayli@net.usj.edu.lb}
\and
\IEEEauthorblockN{Khalil Harris}
\IEEEauthorblockA{khalil.harris@usj.edu.lb}
\and
\IEEEauthorblockN{Jihad Fahs}
\IEEEauthorblockA{jihad.fahs@aub.edu.lb}
\and
\IEEEauthorblockN{Mohamad Assaad}
\IEEEauthorblockA{mohamad.assaad@centralesupelec.fr}
}

\maketitle

\begin{abstract}
Homomorphic encryption (HE) enables privacy-preserving aggregation in federated learning (FL) by allowing the server to operate on encrypted data without decryption. Existing HE-over-the-air (OTA) methods mainly rely on single-key HE schemes and require channel estimation or pre-equalization to compensate for wireless fading. However, single-key HE remains vulnerable to honest-but-curious (HBC) clients holding the shared secret key, while multi-key HE provides stronger client-level security by assigning each device its own secret key. We propose a four-phase protocol that enables the aggregation of xMK-CKKS over a shared wireless channel without channel estimation. The protocol retransmits partial public keys and ciphertexts through the same channel realization, so that the dominant large-modulus encryption terms cancel algebraically during decryption. We integrate this protocol with zero-order FL over slowly varying LoS-dominant channels, where each device transmits a single encrypted scalar per round and the communication/encryption overhead is independent of the model dimension. We show that the residual noise induced by encryption and wireless aggregation preserves the standard convergence rate \(O(1/\sqrt{K})\) up to a negligible noise floor, where $K$ is the number of communication rounds. The protocol assumes a non-trusted server and is secure against HBC clients, preventing any client from recovering the local updates of other participants. Numerical results on MNIST and CIFAR-10 validate the theoretical analysis.
\end{abstract}

\section{Introduction}

Federated learning (FL) \cite{fedavg} enables edge devices to collaboratively train a shared model without sharing their local data. However, the model updates, gradients, or gradient surrogates transmitted by the devices may still reveal sensitive information and remain vulnerable to inference attacks \cite{gradattack}. Over-the-air (OTA) computation \cite{amiri} and zero-order (ZO) gradient estimation \cite{ezofl,mhanna} reduce the uplink communication cost by exploiting the superposition property of the wireless channel. In particular, ZO FL methods \cite{ezofl,mhanna} can reduce the communication load per-device to one or two scalars per round. Nevertheless, these analog transmissions are non-encrypted and therefore do not provide cryptographic protection.

Homomorphic encryption (HE) offers a natural mechanism for privacy-preserving aggregation, since it allows the server to compute directly on encrypted data. Existing HE-over-the-air methods \cite{heairfed,airhe} combine single-key HE with over-the-air aggregation but rely on channel estimation, pre-equalization, or beamforming to compensate for wireless fading. Moreover, because single-key HE places clients in a shared decryption domain, an honest-but-curious (HBC) client holding the shared secret key may decrypt or infer other clients updates.

Multi-key HE addresses this limitation by assigning each device its own secret key. In particular, xMK-CKKS \cite{mkckks} requires all participating devices to contribute partial decryption shares and remains secure against HBC devices. However, secure ring learning with errors (RLWE) based implementations require large ciphertext moduli, e.g., parameter sets such as $(n=4096,q\approx2^{109})$ and $(n=8192,q\approx2^{218})$ in CKKS deployments~\cite{sealcrypto}. At this scale, channel-estimation errors and wireless distortions can leave $q$-scale decryption residuals.

This motivates the central question addressed in this paper: can multi-key HE be combined with over-the-air aggregation without estimating or compensating the wireless channel? We answer this question by proposing a four-phase protocol in which the partial public keys and ciphertexts are transmitted through the same channel realization. As a result, the dominant large-modulus encryption terms experience the same channel coefficients and cancel algebraically during decryption. The remaining channel and encryption noise are not amplified by $q$; instead, they enter the learning algorithm as bounded perturbations.

The proposed protocol is particularly suited to slowly varying LoS-dominant links, such as short-range indoor THz communication scenarios \cite{chen_thz_2024,Jornet2024Evolution}. Such links are often highly directional and dominated by a small number of propagation paths. In fixed or low-mobility deployments, the effective channel can remain approximately constant over the short duration of the protocol phases.

The main contributions of this paper are as follows:
\begin{itemize}
\item We propose a four-phase over-the-air protocol that enables xMK-CKKS aggregation over a shared wireless channel without channel state information (CSI) acquisition or pre-equalization.
\item We show that retransmitting the partial public keys and ciphertexts through the same channel realization causes the dominant $q$-scale encryption terms to cancel algebraically during decryption, avoiding $q$-amplified residual errors.
\item We integrate the protocol with ZO FL over slowly varying LoS-dominant channels, where each device transmits a single encrypted scalar per round and the communication/encryption overhead is independent of the model dimension.
\item We prove that the residual encryption and channel noise preserves the $O(1/\sqrt{K})$ convergence rate, where $K$ is the number of communication rounds, up to a negligible noise floor.
\item We show that the protocol is secure against a non-trusted server and HBC devices. We validate the analysis using MNIST experiments under fading-channel models.
\end{itemize}


The remainder of the paper is organized as follows. Section~II reviews existing HE-over-the-air methods and their limitations. Sections~III and~IV present the system model and the proposed protocol. Section~V provides the convergence analysis results. Section~VI presents numerical results, and Section~VII concludes the paper.

\section{Related Work}
 
 Two recent methods combine HE with over-the-air computation. Wang et al.\ \cite{heairfed} proposed HEAirFed, which uses single-key CKKS: each device encrypts its full gradient vector and the server aggregates over the air using MIMO beamforming with CSI at the devices and the server. Xie et al.\ \cite{airhe} proposed AirHE, a single-key LWE scheme that transmits ciphertext digits as nested-lattice codewords and relies on channel estimation and pre-equalization. Both show that HE can be combined with over-the-air aggregation in single-key settings, but did not address the multi-key setting of this paper, where each device keeps an independent secret key. 
Table~\ref{tab:comparison} shows that the proposed protocol adopts a multi-client architecture, providing stronger resistance against HBC clients and collision-related attacks. Moreover, the scheme follows the Microsoft SEAL recommended RLWE parameters~\cite{sealcrypto}, where the ring dimension is set to $n=4096$ or $n=8192$, values that are considered computationally secure against practical attacks.
 \begin{table}[!ht]
\centering
\caption{Comparison of HE-over-the-air methods.}
\label{tab:comparison}
\footnotesize
\begin{tabular}{@{}lccc@{}}
\toprule
& \textbf{AirHE} \cite{airhe} & \textbf{HEAirFed} \cite{heairfed} & \textbf{Proposed} \\
\midrule
HE scheme & LWE & CKKS & xMK-CKKS \\
Key structure & Single & Single & Multi-key \\
RLWE dim.\ $n$ & $9   $ & $4096$ & $4096, 8192$ \\
Ciph.\ mod.\ $q$ & $6560$ & $2^{30}$ & $2^{109}, 2^{218}$ \\
CSI required & Yes & Yes & No \\
Pre-equal. & Yes & Yes (beamf.) & No \\
Enc.\ payload & Scalar & $\bm{\nabla} F \in\mathbb{R}^d$ & Scalar \\
HBC client & Vuln. & Vuln. & Secure \\
\bottomrule
\end{tabular}

\end{table}
 
\subsection*{Why Existing Methods Do Not Extend to xMK-CKKS} Two limitations separate these methods from the setting of this paper. The first is cryptographic and holds for any modulus: both use single-key HE, so any HBC client holding the shared key can decrypt every other client's update. Single-key HE protects against external adversaries, but does not give resistance to client-level collusion. xMK-CKKS assigns each device its own secret key and requires all devices to contribute decryption shares, which gives resistance against HBC devices.
 
The second is tied to the modulus. Single-key CKKS with beamforming scales to a large modulus without difficulty. The multi-key case is different: the large gap between the modulus $q$ and the scaling factor $\Lambda$ means that any residual channel-estimation error is multiplied by the $q$-scale ciphertext terms and buries the $\Lambda$-scale message (details are presented in Section~\ref{sec:xmk_ckks}). For example, we will show in Fig.~\ref{fig:eq_breaks} of Section~\ref{sec:mumerical_result} that both zero-forcing (ZF) and MMSE pre-equalization diverge from the first iteration at \(q \approx 2^{109}\) . To overcome this limitation, the proposed protocol retransmits the partial public keys and ciphertexts through the same channel realization, allowing the \(q\)-scaled terms to cancel algebraically while only small encryption and channel noise terms remain.
\section{System Model}
Throughout, $i, j \in \{1,\ldots,N\}$ index the devices and $k \in \{0,\ldots,K\}$ the communication rounds.

\subsection{Wireless Zero-Order Federated Learning Model}

Consider an FL framework with $N$ edge devices and a central server coordinating the training of a global model ${\bm \theta}\in\mathbb{R}^d$ over a wireless network. Each device trains on its private local dataset. Let $\mathcal{N}=\{1,\ldots,N\}$ denote the set of devices, and let $F_i:\mathbb{R}^d\to\mathbb{R}$ be the local loss associated with device~$i$. The global objective is to minimize
\begin{equation}\label{eq:global_loss}
F({\bm \theta})
=
\sum_{i=1}^{N} F_i({\bm \theta}),
\qquad
F_i({\bm \theta})
=
\E_{\xi_i\sim\mathcal{D}_i}
\left[
f_i({\bm \theta},\xi_i)
\right],
\end{equation}
where $\xi_i$ is sampled from the local data distribution $\mathcal{D}_i$. The functions $F$, $F_i$, and $f_i$ are allowed to be nonconvex.

Let $h_{i,k}$ denote the channel coefficient between device~$i$ and the server during communication round~$k$. We consider slowly varying block-fading channels with a nonzero mean due to a LoS component:
\begin{equation}\label{eq:channel_stats}
\E[h_{i,k}] = \mu_i \neq 0,
\qquad
\E[h_{i,k}^2] = \Omega_i,
\qquad
1\leq i\leq N.
\end{equation}
The channel coefficients are assumed independent, but not necessarily identically distributed, across devices. Within each round~$k$, $h_{i,k}$ remains constant for each device~$i$ over all protocol phases. Thus, each round corresponds to one fading block, whereas independent fading is assumed from one round to the next. This model is motivated by short-range indoor THz links with fixed or low-mobility devices. Such links are often LoS-dominant and highly directional, with sparse multipath and high Rician $K$-factors. When the transmitter, receiver, and dominant scatterers remain stationary over the duration of a communication round, the coherence time can cover all phases of the proposed protocol~\cite{thz_indoor,chen_thz_2021}.

We consider a ZO method in which the channel disturbance is incorporated into the learning process, in the same spirit as~\cite{ezofl}. At each round~$k$, every device~$i$ computes the standard two-point ZO difference of its local loss~\cite{duchi,agarwal,ezofl}:
\begin{equation}\label{eq:delta_f}
\Delta f_{i,k}
=
f_i({\bm \theta}_k+\gamma_k{\bm \Phi}_k,\xi_{i,k})
-
f_i({\bm \theta}_k-\gamma_k{\bm \Phi}_k,\xi_{i,k}),
\end{equation}
where ${\bm \Phi}_k=(\Phi_k^1,\ldots,\Phi_k^d)^\top$ is a perturbation vector with i.i.d.\ entries satisfying $\E[(\Phi_k^j)^2]=b_1$ and $\|{\bm \Phi}_k\|\leq b_2$, and $\gamma_k > 0$ is the smoothing parameter. The perturbation sequence is generated randomly and is made available to all devices.

When the channel has a nonzero mean, $\E[h_{i,k}]=\mu_i\neq0$, in the ZO FL scheme~\cite{ezofl}, 
each device knows, or estimates, the long-term channel mean $\mu_i$ and transmits each round the single scalar $\Delta f_{i,k}/\mu_i$. The server receives
\begin{equation}\label{eq:Y2_ezofl}
Y_k
=
\sum_{i\in\mathcal{N}}
\frac{\Delta f_{i,k}}{\mu_i}h_{i,k}
+
n_k,
\end{equation}
where $n_k\sim\mathcal{N}(0,\sigma_n^2)$, and broadcasts $Y_k$ to all devices. Each device then forms the ZO gradient estimate
\begin{equation}\label{eq:gk_ezofl}
{\bm g}_k
=
{\bm \Phi}_kY_k.
\end{equation}
The channel coefficient $h_{i,k}$ is not estimated or removed; it enters the gradient estimate as a perturbation. In~\cite{ezofl}, it is shown that including the channel disturbance in the learning process does not change the convergence rate, which remains of order $O(1/\sqrt{K})$ in nonconvex settings. In Theorem~\ref{thm:th2}, we extend the analysis to the encrypted setting and show that the proposed protocol preserves the same rate up to an explicit noise floor, namely $O(1/\sqrt{K})+\rho$, where $\rho$ is negligible under the considered parameter regime.

\subsection{xMK-CKKS}
\label{sec:xmk_ckks}
xMK-CKKS is considered one of the most prominent and suitable HE schemes for FL scenarios, as it provides multi-key security, supports floating-point computations, and relies on the RLWE lattice problem, for which no practical attacks are currently known under recommended parameters.\\
The xMK-CKKS scheme~\cite{mkckks} operates over the ring
\(
R_q=\mathbb{Z}_q[X]/(X^n+1),
\)
where \(n\) denotes the polynomial ring degree (RLWE dimension) and \(q\) represents the ciphertext coefficient modulus.\\
The encoded message uses a scaling factor \(\Lambda\) to balance numerical precision and noise growth.
Different cryptographic components and parameters of the xMK-CKKS scheme are summarized in the steps below:
\begin{itemize}

\item \textbf{Secret-key generation.} Each device $d_i$ independently samples its secret key \(
s_i \in \{-1,1\}^n\) according to Bernoulli(1/2). 
\item \textbf{Partial public-key generation.} Each device computes a partial public key
\(
b_i=-s_i a+e_i \mod q,
\)
where $a$ is a common public polynomial and
\(
e_i\sim\mathcal{N}(0,\sigma_e^2)^{\otimes n},
\)
with $\sigma_e=3.2$.

\item \textbf{Aggregated public-key construction.} The partial public keys are summed to form the aggregated public key
\(
\tilde{b}=\sum_{i=1}^{N}b_i,
\)
which gives 
\begin{equation}\label{eq:cancel_digital}
\tilde{b}+a\sum_{i=1}^{N}s_i
=
\tilde{b}+aS
=
\sum_{i=1}^{N}e_i
:=
E_{agg},
\end{equation}
where
\(
S:=\sum_{i=1}^{N}s_i\). 

\item \textbf{Encryption.} Each device encrypts its plaintext $m_i$ using the aggregated public key
\(
ct_i=(c_{0,i},c_{1,i}),
\)
where
\(
c_{0,i}=v_i\tilde{b}+m_i+e_i^{(0)},
\)
and
\(
c_{1,i}=v_i a+e_i^{(1)} \mod q
\) , where $v_i$ is a fresh random polynomial used to randomize the encryption process.

\item \textbf{Homomorphic aggregation.} The server aggregates the ciphertexts
\(
C_{\mathrm{sum}}
=
\sum_{i=1}^{N}ct_i
=
(C_{\mathrm{sum},0},C_{\mathrm{sum},1}).
\)

\item \textbf{Partial decryption.} Each device computes a decryption share
\(
D_i=s_i C_{\mathrm{sum},1}+e_i^{*}\mod q,\) where $e_i^{*}$ is the effective decryption-noise term associated with device $d_i$, modeled as a discrete Gaussian random variable with variance determined by the underlying encryption errors.
\item \textbf{Collaborative reconstruction.} The server combines the aggregated ciphertext and all decryption shares
\(
C_{\mathrm{sum},0}+\sum_{i=1}^{N}D_i \mod q.
\)
After expansion, the large key-dependent terms involving $a$ cancel algebraically, leaving
\(
\sum_{i=1}^{N}m_i+\text{bounded noise}.
\)
Therefore,
\(
C_{\mathrm{sum},0}+\sum_{i=1}^{N}D_i
\approx
\sum_{i=1}^{N}m_i.
\)
\end{itemize}
A detailed correctness analysis of the collaborative decryption procedure in xMK-CKKS is provided in Appendix A.

\subsection{Why Naive OTA xMK-CKKS Fails}

We now explain why directly reusing digitally generated xMK-CKKS public keys over a wireless channel fails. Suppose first that, in round~$k$, the partial public keys are transmitted over the air through the same channel realization $h_{i,k}$. The server receives
\begin{align*}
\tilde{b}_k
&=
\sum_{i=1}^N h_{i,k} b_i + w_{b,k} \nonumber\\
&=
-\left(\sum_{i=1}^N h_{i,k}s_i\right)a
+
\sum_{i=1}^N h_{i,k}e_i
+
w_{b,k}.
\end{align*}
Define
\begin{equation}
\tilde S_k:=\sum_{i=1}^N h_{i,k}s_i,
\qquad
\tilde E_{\mathrm{agg},k}
:=
\sum_{i=1}^N h_{i,k}e_i+w_{b,k}.
\label{eq:def}
\end{equation}
Then the channel-weighted cancellation identity becomes
\begin{equation}\label{eq:cancel_ota}
\tilde{b}_k
+
\tilde S_k a
=
\tilde E_{\mathrm{agg},k},
\end{equation}
where the large $q$-scale term cancels, since $\tilde b_k$ and $\tilde S_k$ are generated with the same channel realization.

The failure occurs if the public keys are aggregated once during an initial setup phase and then reused in later rounds. In that case, the server receives $\tilde b_0$, which is tied to the setup channel realization $h_{i,0}$, whereas the decryption shares in round~$k$ produce $\tilde S_k$, which is tied to the current channel realization $h_{i,k}$. The cancellation becomes
\begin{equation*}\label{eq:residual_full}
\tilde b_0
+
\tilde S_k a
=
(\tilde S_k-\tilde S_0)a
+
\tilde E_{\mathrm{agg},0}.
\end{equation*}
where $\tilde S_0 := \sum_{i=1}^N h_{i,0}s_i$ is the channel-weighted secret from the setup round, so that $\tilde b_0 = -\tilde S_0 a + \tilde E_{\mathrm{agg},0}$. 
The first term on the right-hand side is no longer an encryption-noise term. It is a channel-mismatch residual multiplied by the public polynomial $a\in R_q$. Since the remaining coefficients are reduced modulo \(q\), with \(q \approx 2^{109}\) or \(q \approx 2^{218}\) following the Microsoft SEAL recommendations~\cite{sealcrypto}, the residual term may reach a magnitude proportional to \(q\).

By contrast, the encoded message has a scale of approximately $\Lambda$, with $\Lambda=2^{40}$. Thus, a $q$-scale residual can be roughly $q/\Lambda \approx 2^{70}$ times larger than the encoded signal, which destroys decryption. The proposed protocol avoids this failure by retransmitting the partial public keys in every communication round, through the same channel realization that carries the encrypted scalar and the decryption shares. The proposed protocol does not try to estimate or invert the wireless channel; instead, it forces the public-key and decryption-share terms to experience the same fading coefficients, so that the large-modulus terms cancel algebraically. The complete protocol is described in the next section.

\section{Proposed Protocol}
\label{sec:pp}

\begin{figure*}[t]
\centering
\begin{tikzpicture}[>=stealth, font=\scriptsize, yscale=1.05]

\node[font=\footnotesize\bfseries, text=orange!60!black] at (7.4,0.42) {OTA Aggregation};

\draw[rounded corners=6pt, fill=gray!8, draw=black!60, line width=1pt] (11.0,0.5) rectangle (15.8,-5.0);
\draw[line width=1.2pt, black!70] (13.4,0.5) -- (13.4,1.05);
\draw[line width=0.8pt, black!70] (13.15,0.5) -- (13.4,1.05) -- (13.65,0.5);
\fill[black!70] (13.4,1.05) circle (0.04);
\draw[black!40, line width=0.6pt] (13.05,1.1) arc[start angle=200, end angle=340, radius=0.22];
\draw[black!40, line width=0.6pt] (12.95,1.2) arc[start angle=200, end angle=340, radius=0.32];
\node[font=\normalsize\bfseries] at (13.4,0.15) {Server};

\node[align=left, anchor=west, font=\scriptsize] at (11.2,-0.2) {%
Receives $\tilde{b}_k$};
\node[align=left, anchor=west, font=\scriptsize] at (11.2,-0.45) {%
Broadcasts $\tilde{b}_k$ to devices};
\draw[gray!40] (11.0,-0.75) -- (15.8,-0.75);

\node[align=left, anchor=west, font=\scriptsize] at (11.2,-1.4) {%
Receives $(\tilde{c}_{0,k}, \tilde{c}_{1,k})$};
\node[align=left, anchor=west, font=\scriptsize] at (11.2,-1.65) {%
Broadcasts $\tilde{c}_{1,k}$ to devices};
\draw[gray!40] (11.0,-1.95) -- (15.8,-1.95);

\node[align=left, anchor=west, font=\scriptsize] at (11.2,-2.7) {%
Receives $\tilde{D}_k$};
\draw[gray!40] (11.0,-3.0) -- (15.8,-3.0);

\node[align=left, anchor=north west, font=\scriptsize, text width=4.2cm] at (11.2,-3.1) {%
\textbf{Recovery:}\\
$\hat{M}_k = \tilde{c}_{0,k} + \tilde{D}_k$\\
$\tilde{b}_k + \tilde{S}_k a = \tilde{E}_{\text{agg},k}$\;\textit{(q cancels)}\\
$M_k = \sum_i \frac{h_{i,k}}{\mu_i}\Delta f_{i,k} + \varepsilon_k$\\
Broadcasts $M_k$ to devices};

\fill[gray!8, rounded corners=5pt] (0.55,0.18) rectangle (7.35,-4.08);
\draw[rounded corners=5pt, draw=gray!50, line width=0.6pt] (0.55,0.18) rectangle (7.35,-4.08);
\node[font=\footnotesize\bfseries, text=gray!60] at (2.1, 0.35) {Edge Devices};

\node[font=\small\bfseries, anchor=east] at (0.5, -0.2) {Phase 1};
\node[font=\tiny, anchor=east, text=black!50] at (0.5, -0.45) {Key Aggr.};

\draw[rounded corners=2pt, fill=white, draw=blue!40] (0.7, 0.05) rectangle (3.5,-0.12) node[midway] {Device 1: sends $b_1$};
\draw[rounded corners=2pt, fill=white, draw=blue!40] (0.7,-0.2) rectangle (3.5,-0.37) node[midway] {Device 2: sends $b_2$};
\draw[rounded corners=2pt, fill=white, draw=blue!40] (0.7,-0.45) rectangle (3.5,-0.62) node[midway] {Device $N$: sends $b_N$};

\draw[rounded corners=4pt, fill=orange!8, draw=orange!40] (5.3, 0.1) rectangle (9.5,-0.68);
\node at (7.4, -0.29) {$\tilde{b}_k = \sum_i h_{i,k} b_i + w_{b,k}$};

\draw[->, thick, red!60!black] (3.5,-0.04) -- (5.3,-0.15);
\draw[->, thick, red!60!black] (3.5,-0.29) -- (5.3,-0.29);
\draw[->, thick, red!60!black] (3.5,-0.54) -- (5.3,-0.43);
\draw[->, thick, red!60!black] (9.5,-0.29) -- (11.0,-0.29);

\draw[->, thick, dashed, blue!50!black] (11.0,-0.92) -- (7.35,-0.92) node[midway, below, font=\scriptsize, fill=white, inner sep=1pt] {broadcast $\tilde{b}_k$};
\draw[gray!30, dashed] (0.0,-1.08) -- (16.0,-1.08);

\node[font=\small\bfseries, anchor=east] at (0.5, -1.5) {Phase 2};
\node[font=\tiny, anchor=east, text=black!50] at (0.5, -1.75) {Encryption};

\draw[rounded corners=2pt, fill=white, draw=blue!40] (0.7,-1.20) rectangle (4.5,-1.40) node[midway] {Dev.\ 1: encrypts $\Delta f_{1,k}/\mu_i$};
\draw[rounded corners=2pt, fill=white, draw=blue!40] (0.7,-1.48) rectangle (4.5,-1.68) node[midway] {Dev.\ 2: encrypts $\Delta f_{2,k}/\mu_i$};
\draw[rounded corners=2pt, fill=white, draw=blue!40] (0.7,-1.76) rectangle (4.5,-1.96) node[midway] {Dev.\ $N$: encrypts $\Delta f_{N,k}/\mu_i$};

\draw[rounded corners=4pt, fill=orange!8, draw=orange!40] (5.3,-1.17) rectangle (9.5,-1.94);
\node at (7.4, -1.56) {$(\tilde{c}_{0,k},\;\tilde{c}_{1,k})$};

\draw[->, thick, red!60!black] (4.5,-1.30) -- (5.3,-1.38);
\draw[->, thick, red!60!black] (4.5,-1.58) -- (5.3,-1.56);
\draw[->, thick, red!60!black] (4.5,-1.86) -- (5.3,-1.74);
\draw[->, thick, red!60!black] (9.5,-1.56) -- (11.0,-1.56);

\draw[->, thick, dashed, blue!50!black] (11.0,-2.18) -- (7.35,-2.18) node[midway, below, font=\scriptsize, fill=white, inner sep=1pt] {broadcast $\tilde{c}_{1,k}$};
\draw[gray!30, dashed] (0.0,-2.33) -- (16.0,-2.33);

\node[font=\small\bfseries, anchor=east] at (0.5, -2.75) {Phase 3};
\node[font=\tiny, anchor=east, text=black!50] at (0.5, -3.0) {Dec.\ Share};

\draw[rounded corners=2pt, fill=white, draw=blue!40] (0.7,-2.48) rectangle (3.5,-2.65) node[midway] {Device 1: sends $D_1$};
\draw[rounded corners=2pt, fill=white, draw=blue!40] (0.7,-2.73) rectangle (3.5,-2.90) node[midway] {Device 2: sends $D_2$};
\draw[rounded corners=2pt, fill=white, draw=blue!40] (0.7,-2.98) rectangle (3.5,-3.15) node[midway] {Device $N$: sends $D_N$};

\draw[rounded corners=4pt, fill=orange!8, draw=orange!40] (5.3,-2.43) rectangle (9.5,-3.2);
\node at (7.4, -2.82) {$\tilde{D}_k = \sum_j h_{j,k} D_j + w_{D,k}$};

\draw[->, thick, red!60!black] (3.5,-2.57) -- (5.3,-2.62);
\draw[->, thick, red!60!black] (3.5,-2.82) -- (5.3,-2.82);
\draw[->, thick, red!60!black] (3.5,-3.07) -- (5.3,-3.02);
\draw[->, thick, red!60!black] (9.5,-2.82) -- (11.0,-2.82);

\draw[gray!30, dashed] (0.0,-3.4) -- (16.0,-3.4);

\node[font=\small\bfseries, anchor=east] at (0.5, -3.75) {Phase 4};
\node[font=\tiny, anchor=east, text=black!50] at (0.5, -4.0) {Update};

\draw[rounded corners=2pt, fill=white, draw=blue!40] (0.7,-3.55) rectangle (7.2,-3.95);
\node at (3.95, -3.75) {All devices: ${\bm g}_k = {\bm \Phi}_k M_k$, \quad ${\bm \theta}_{k+1} = {\bm \theta}_k - \eta_k {\bm g}_k$};

\draw[->, thick, dashed, blue!50!black] (11.0,-3.75) -- (7.35,-3.75) node[midway, below, font=\scriptsize, fill=white, inner sep=1pt] {broadcast $M_k$};

\draw[rounded corners=3pt, fill=white, draw=gray!40] (3.0,-4.4) rectangle (9.5,-4.9);
\draw[->, thick, red!60!black] (3.2,-4.55) -- (4.2,-4.55);
\node[anchor=west, font=\scriptsize] at (4.3,-4.55) {Uplink (same resource, different channels)};
\draw[->, thick, dashed, blue!50!black] (3.2,-4.75) -- (4.2,-4.75);
\node[anchor=west, font=\scriptsize] at (4.3,-4.75) {Downlink (broadcast)};

\end{tikzpicture}
\caption{The proposed four-phase protocol for one round~$k$. All devices transmit simultaneously on the same wireless resource. Each device holds its own secret key~$s_i$. The server does not estimate any channel coefficient. In Phase~4, the $q$-scale terms cancel algebraically because $\tilde{b}_k$ and $\tilde{S}_k$ use the same channel realization~$h_{i,k}$.}
\label{fig:protocol}
\end{figure*}
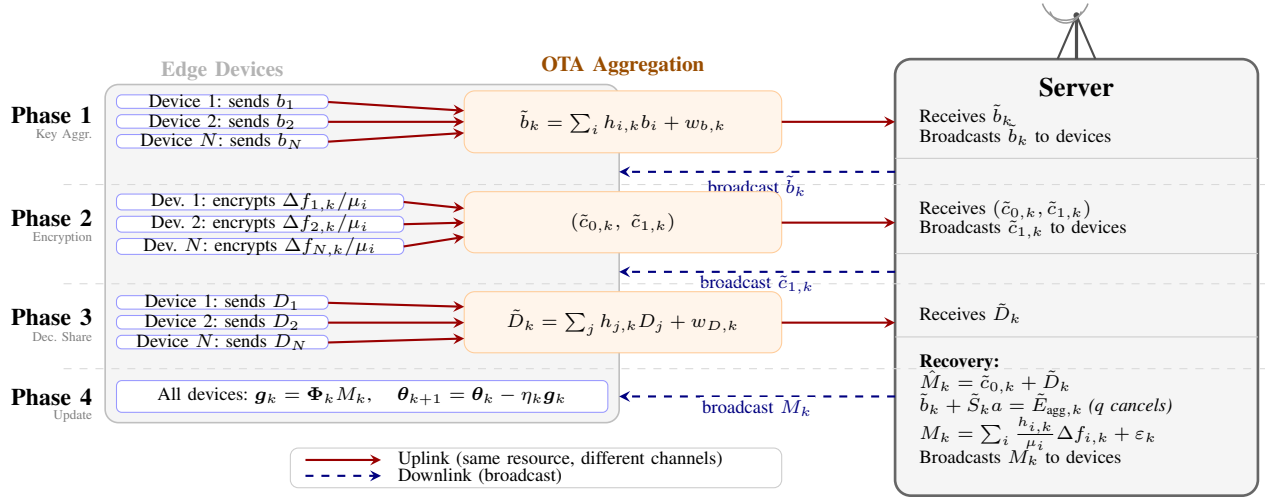

The key observation is that the cancellation identity present in Eq.~\eqref{eq:cancel_digital} fails over the air because the public keys $b_i$ are aggregated once during setup, while the channel coefficients $h_{i,k}$ change every round. To fix this, we retransmit $b_i$ every round through the same channel realization $h_{i,k}$ that carries the ciphertext and decryption shares. This ensures that the cancellation identity survives in channel-weighted form.

The protocol consists of four phases per round, illustrated in Fig.~\ref{fig:protocol}. We add the following assumption on the HE randomness.

\noindent \textbf{Assumption~1.} (HE randomness): The fresh randomness $v_{i,k}$, $e_{0,i,k}$, $e_{1,i,k}$, $e_{j,k}^*$ and the channel noises $w_{b,k}$, $w_{0,k}$, $w_{1,k}$, $w_{D,k}$ are mutually independent, zero-mean, and independent of ${\bm \Phi}_k$, $h_{\cdot,k}$, $\xi_{\cdot,k}$.

We assume a perfect downlink: messages broadcasted by the server 
are received by all devices without errors. This is a standard assumption in OTA FL \cite{amiri,ezofl}.

\subsection{Phase~1: Fresh Key Aggregation}

Each client~$i$ transmits its partial public key $b_i$ over the uplink. The server receives:
\begin{equation}\label{eq:bk}
\tilde{b}_k = \sum_{i=1}^{N} h_{i,k} b_i + w_{b,k} = -\tilde{S}_k a + \tilde{E}_{\mathrm{agg},k},
\end{equation}
where $\tilde{S}_k$ 
and $\tilde{E}_{\mathrm{agg},k}$ 
are given by Eq.~\eqref{eq:def}, and the channel-weighted cancellation identity~\eqref{eq:cancel_ota} holds.

\subsection{Phase~2: Encryption}

The server broadcasts $\tilde{b}_k$ to all clients. Each client~$i$ encodes its ZO estimate as $m_{i,k} = \lfloor \Lambda \Delta f_{i,k}/\mu_i \rceil$ and transmits the ciphertext pair $(c_0^{(i)}, c_1^{(i)})$:
\begin{equation}
c_0^{(i)} = v_{i,k}\tilde{b}_k + m_{i,k} + e_{0,i,k}, \quad c_1^{(i)} = v_{i,k} a + e_{1,i,k}, \label{eq:c01i}
\end{equation}
where $v_{i,k}$ denotes a fresh encryption randomness, while $e_{0,i,k}$ and $e_{1,i,k}$ represent small error polynomials sampled according to a Gaussian distribution, as described in Section~\ref{sec:xmk_ckks}. The server receives the aggregated ciphertexts over the air:
\begin{align}
\tilde{c}_{0,k} &= V_k \tilde{b}_k + \sum_i h_{i,k} m_{i,k} + \sum_i h_{i,k} e_{0,i,k} + w_{0,k}, \label{eq:c0agg}\\
\tilde{c}_{1,k} &= V_k a + \sum_i h_{i,k} e_{1,i,k} + w_{1,k}, \label{eq:c1agg}
\end{align}
where $V_k := \sum_i h_{i,k} v_{i,k}$.

\subsection{Phase~3: Decryption Share}

The server broadcasts $\tilde{c}_{1,k}$. Each client~$j$ computes and transmits a partial decryption share $D_{j,k} = s_j \tilde{c}_{1,k} + e_{j,k}^*$, where $s_j$ denotes the secret key of client $j$ and $e^{*}_{j,k}$ is an additional Gaussian noise term used to protect the privacy of the partial decryption share. The server receives:
\begin{equation}\label{eq:Dk}
\tilde{D}_k = \tilde{S}_k \tilde{c}_{1,k} + \sum_j h_{j,k} e_{j,k}^* + w_{D,k}.
\end{equation}

\subsection{Phase~4: Recovery}

The server forms $\hat{M}_k := \tilde{c}_{0,k} + \tilde{D}_k$. Substituting Eq.~\eqref{eq:c0agg}, Eq.~\eqref{eq:c1agg}, and Eq.~\eqref{eq:Dk}:
\begin{align}
\hat{M}_k &= V_k \tilde{b}_k + \tilde{S}_k(V_k a + \sum_i h_{i,k} e_{1,i,k} + w_{1,k}) \notag\\
&\quad + \sum_i h_{i,k} m_{i,k} + \sum_i h_{i,k} e_{0,i,k} + w_{0,k} \notag\\
&\quad + \sum_j h_{j,k} e_{j,k}^* + w_{D,k} \notag\\
&= V_k(\tilde{b}_k + \tilde{S}_k a) + \sum_i h_{i,k} m_{i,k} + \zeta_k, 
\label{eq:Mhat_expand}
\end{align}
where 

\begin{multline}\label{eq:zeta}
\zeta_k := \tilde{S}_k\sum_i h_{i,k} e_{1,i,k} + \tilde{S}_k w_{1,k}
+ \sum_i h_{i,k} e_{0,i,k} + w_{0,k}\\
+ \sum_j h_{j,k} e_{j,k}^* + w_{D,k}.
\end{multline}

\subsection{Ring-Scale Cancellation}

The term $V_k(\tilde{b}_k + \tilde{S}_k a)$ contains two $q$-scale components $V_k\tilde{b}_k$ and $\tilde{S}_k V_k a$, each of order $q \approx 2^{109}$ or $2^{218}$, far larger than the message $\Lambda\sum_i h_{i,k}\Delta f_{i,k}$ of order $\Lambda \approx 2^{40}$. Substituting $\tilde{b}_k = -\tilde{S}_k a + \tilde{E}_{\mathrm{agg},k}$ from Eq.~\eqref{eq:bk} gives $V_k(\tilde{b}_k + \tilde{S}_k a) = V_k\tilde{E}_{\mathrm{agg},k}$: the $q$-scale terms cancel exactly, leaving a product of two small polynomials whose second moment is of order $n^2\sigma_e^2\big(\sum_i\Omega_i\big)^2$ (more details are found in Appendix B), independent of $q$. The remaining expression is:
\begin{equation*}\label{eq:Mhat_final}
\hat{M}_k = \Lambda \sum_i \frac{h_{i,k}}{\mu_i}\Delta f_{i,k} + \hat{\varepsilon}_k,
\end{equation*}
where


\begin{equation}\label{eq:varepshat}
\hat{\varepsilon}_k = V_k \tilde{E}_{\mathrm{agg},k} + \zeta_k,
\end{equation}

where $\zeta_k$ is given in Eq.~\eqref{eq:zeta}.
 
 Decoding by $1/\Lambda$:
\begin{equation}\label{eq:Mk}
M_k = \frac{1}{\Lambda}\hat{M}_k = \sum_i \frac{h_{i,k}}{\mu_i} \Delta f_{i,k} + \varepsilon_k, \quad \varepsilon_k := \frac{1}{\Lambda}\hat{\varepsilon}_k.
\end{equation}

Each ring coefficient is an integer in $\{0,\ldots,q-1\}$, transmitted as an analog symbol; rounding the real-valued superposition recovers the correct ring element, and the algebraic identities of Phases~1-4 hold under modular arithmetic, when the noise $\hat{\varepsilon}$ is small with high probability~\cite{mkckks}.

\subsection{Security of the OTA Protocol}

Retransmitting $b_i$ every round raises the question of whether the server can extract individual $b_i$ from the $K$ superpositions $\tilde{b}_k = \sum_i h_{i,k} b_i + w_{b,k}$. The system has $K$ observations and $KN + N$ unknowns (channel coefficients and partial keys) and is under-determined since the server never observes individual $h_{i,k}$. Even if an adversary recovers each \(b_i\), this information is already publicly available in digital xMK-CKKS. Recovering the corresponding secret key \(s_i\) from
$b_i = -s_i a + e_i$ remains computationally hard under the RLWE assumption, whose security is reducible to hard lattice problems such as the Closest Vector Problem (CVP)~\cite{mkckks}. The OTA setting therefore provides two layers of protection: the channel layer prevents the server from separating individual transmissions, while the cryptographic layer protects each $b_i$ even if it were recovered. In the digital setting only the second layer is present.

\subsection{Gradient Estimator}
Each device forms the gradient estimate from the $M_k$:
\begin{equation}\label{eq:gk_enc}
{\bm g}_k = {\bm \Phi}_k \, M_k = {\bm \Phi}_k\left(\sum_i \frac{h_{i,k}}{\mu_i} \Delta f_{i,k} + \varepsilon_k\right),
\end{equation}
matching the structure of Eq.~\eqref{eq:gk_ezofl} with the channel noise $n_k$ replaced by the decoded HE noise $\varepsilon_k$. The model update is ${\bm \theta}_{k+1} = {\bm \theta}_k - \eta_k {\bm g}_k$.
\subsection{Algorithm}
Algorithm~\ref{algo:Enc_EZOFL} summarizes the main steps of the proposed EEZOFL protocol based on xMK-CKKS.
\begin{algorithm}[H]
\caption{EEZOFL}
\textbf{Input:} initial values ${\bm \theta}_0$, $\eta_0$, $\gamma_0$, channel mean $\mu_i$
\begin{algorithmic}[1]
\FOR{$k = 0, \ldots, K$}
\STATE Each device~$i$ computes $\Delta f_{i,k}$ using its local data and encodes $m_{i,k} = \lfloor \Lambda \Delta f_{i,k}/\mu_i \rceil$.
\STATE \textit{Phase~1:} Each device transmits $b_i$. The server receives $\tilde{b}_k$ given in Eq.~\eqref{eq:bk} and broadcasts it.
\STATE \textit{Phase~2:} Each device transmits $(c_0^{(i)}, c_1^{(i)})$ given in Eq.~\eqref{eq:c01i}. The server receives $(\tilde{c}_{0,k}, \tilde{c}_{1,k})$ and broadcasts $\tilde{c}_{1,k}$.
\STATE \textit{Phase~3:} Each device~$j$ transmits $D_{j,k} = s_j \tilde{c}_{1,k} + e_{j,k}^*$. The server receives $\tilde{D}_k$ given in Eq.~\eqref{eq:Dk}.
\STATE \textit{Phase~4:} The server computes $M_k$ given in Eq.~\eqref{eq:Mk} and broadcasts $M_k$ to all devices.
\STATE Each device multiplies the received value by ${\bm \Phi}_k$ to obtain ${\bm g}_k$ given in Eq.~\eqref{eq:gk_enc}.
\STATE Each device updates the model ${\bm \theta}_{k+1} = {\bm \theta}_k - \eta_k {\bm g}_k$.
\ENDFOR
\end{algorithmic}
\label{algo:Enc_EZOFL}
\end{algorithm}
\subsection{Encryption Overhead and Feasibility}
\label{sec:Enc_Over_Head}
Table~\ref{tab:he_params} presents two secure xMK-CKKS parameter sets, corresponding to \((n, log_{2}q)=(4096,109)\) and \((n, log_{2}q)=(8192,218)\) , following the HE Standard~\cite{HEStandard} and Microsoft SEAL~\cite{sealcrypto} recommendations, both providing an estimated \(128\)-bit classical security level.\\
The uplink communication per round is $4n\lceil\log_2 q\rceil$ bits per device. 
This includes the transmission of the aggregated public key $\tilde{b}_k$ in Phase~1, 
the two ciphertext components $(c_0^{(i)},c_1^{(i)})$ in Phase~2, and the partial 
decryption share $D_{i,k}$ in Phase~3. Therefore, the transmission time per device 
at bandwidth $B$ is given by
\(
T_{\mathrm{tx}}=\frac{4n\lceil\log_2 q\rceil}{B}.
\)
Hence the transmission time reported in Table~\ref{tab:he_params} corresponds to 
the communication performed during the first three phases of the proposed 
protocol described in Section~\ref{sec:pp}. Therefore, increasing the security 
level also increases the communication overhead, which directly impacts the 
transmission time. At $B=1$~THz, both configurations transmit in less than 
$8\mu~s$, which remains well within the coherence time of static indoor THz 
channels.

The encryption benchmarks were obtained on an Intel Core i9-14900HX platform (24 cores, 32 threads, 32\,GB DDR5 RAM). Increasing the parameter set from $(n,\log_2 q)=(4096,109)$ to $(8192,218)$ raises the average encryption time from 23.92\,ms to 55.56\,ms due to the increased complexity of polynomial arithmetic over larger rings. Despite this increase, the encryption overhead remains practical for FL applications.
From a computational perspective, encryption and decryption-share operations have a complexity of $O(n\log n)$ per device per round due to the use of the number theoretic transform (NTT). This explains the increase in encryption latency when moving from $(n,\log_2 q)=(4096,109)$ to $(8192,218)$. In addition, the decoded noise variance scales as $\bar{\sigma}_\varepsilon^2 \propto n/\Lambda^2$ as shown in Appendix~B, while the resulting noise floor $\rho$ remains negligible for both sets of parameters.

Finally, as shown in Table~\ref{tab:he_params}, the ciphertext size increases from
$109$~KB for $n=4096$ to $446$~KB for $n=8192$, while the plaintext is represented by a floating-point value of only $8$ bytes. This corresponds to communication expansion factors of approximately $1.36\times10^4$ and $5.58\times10^4$, respectively. Although significant, such an increase is inherent to HE-based systems, where higher security levels and larger parameter sets lead to larger ciphertexts. The results therefore highlight the classical tradeoff between security and communication efficiency in HE.\\

\begin{table}[!ht]
\centering
\caption{HE parameter sets and protocol overhead.}
\label{tab:he_params}
\small
\begin{tabular}{@{}lcc@{}}
\toprule
& $n\!=\!4096$ & $n\!=\!8192$ \\
\midrule
$\log_2 q$ & 109 & 218 \\
$\Lambda$ & $2^{40}$ & $2^{40}$ \\
Security (HE Std.) & $128$-bit & $128$-bit \\
Ciphertext size & 109~KB & 446~KB \\
Encryption time (mean) & 23.92~ms &
55.56~ms\\	
Uplink / round & 1.8~Mbit & 7.1~Mbit \\
$T_{\mathrm{tx}}$ ($B\!=\!1$~THz) & 1.8~$\mu$s & 7.1~$\mu$s \\
Storage / device & 109~KB & 446~KB \\
Encrypt (NTT) & \multicolumn{2}{c}{$O(n\log n)$ per round} \\
$\rho$ (noise floor) & $\sim\!2\!\times\!10^{-16}$ & $\sim\!4\!\times\!10^{-16}$ \\
\bottomrule
\end{tabular}
\end{table}

\section{Convergence Analysis}
\label{sec:Con_Analysis}

We retain Assumptions~3.1--3.4 of~\cite{ezofl}: $L$-smoothness of the global objective function $F(\cdot)$ and bounded Hessian $\|\nabla^2 F_i\|_2 \leq b$ (Assumption~3.1), Lipschitz continuity of $f_i(\cdot, \xi_i)$ with constant $L$ (Assumption~3.2), step-size conditions $\sum_k \eta_k\gamma_k = \infty$, $\sum_k \eta_k\gamma_k^3 < \infty$, $\sum_k \eta_k^2\gamma_k^2 < \infty$ (Assumption~3.3), and perturbation vector ${\bm \Phi}_k$ with i.i.d.\ entries satisfying $\E[({\Phi}_k^j)^2] = b_1$ and $\|{\bm \Phi}_k\| \leq b_2$ (Assumption~3.4). We assume a perfect downlink: the server broadcasts $\tilde{b}_k$, $\tilde{c}_{1,k}$, and $M_k$ to all devices without error. Let $\calH_k = \{{\bm \theta}_0, \xi_0, \ldots, {\bm \theta}_{k-1}, \xi_{k-1}, {\bm \theta}_k\}$ denote the history up to and including the model ${\bm \theta}_k$ but excluding the current sample $\xi_k$.

\subsection{Preliminary results}
\label{sec_Prem_results}
\begin{proposition}(Noise budget)
\label{prop:prop1}
Under Assumption~1 and the system model in~\eqref{eq:Y2_ezofl}, there exists a $B^2_\varepsilon$  such that $\hat{\varepsilon}_k$ defined in Eq.~\eqref{eq:varepshat} satisfies
\begin{equation}\label{eq:noise_budget_prop}
\E[\hat{\varepsilon}_k] = 0, \quad \E[\|\hat{\varepsilon}_k\|^2 \mid \calH_k] \leq B_\varepsilon^2,
\end{equation}
where the expression of $B_\varepsilon^2$ is given in Appendix~B. The decryption correctness holds with probability greater than $1 - 4B_\varepsilon^2/q^2$. The per-coefficient decoded noise variance is:
\begin{equation}\label{eq:sigma_eps_main}
\bar{\sigma}_\varepsilon^2 := \frac{B_\varepsilon^2}{n\Lambda^2}
= O\!\left(\frac{n\,\sigma_e^2\big(\sum_{i=1}^{N}\Omega_i\big)^2}{\Lambda^2}\right).
\end{equation}
\end{proposition}
\begin{proof}
See Appendix~B.
\end{proof}
The noise budget depends on the channel statistics only through the aggregate sums $\sum_i \Omega_i$, $\sum_i \Omega_i^2$, and $\sum_i \mu_i^4$; For example, in the case of identical channels, $\Omega_i = \Omega$ for all $i$ and Eq.~\eqref{eq:sigma_eps_main} reduces to $\bar{\sigma}_\varepsilon^2 \approx \{6\times10^{-17}\,\Omega^2, 3\times10^{-17}\,\Omega^2\}$  for the parameters of Table~\ref{tab:he_params} corresponding to $n = \{4096,8192\}$ respectively. In both cases, the decoded noise variance is more than sixteen orders of magnitude below typical channel noise, so the encryption layer is invisible to the learning algorithm.

\begin{lemma} (Bias and second moment) 
\label{lem:lem1}
Under Assumptions~3.1--3.4 and Assumption~1, 
\begin{eqnarray*}
\E[{\bm g}_k \mid \calH_k] &=& c_1\gamma_k(\nabla F({\bm \theta}_k) + {\bm \delta}_k)\nonumber\\
\E[\|{\bm g}_k\|^2 \mid \calH_k] &\leq& \tilde{C}'_\gamma\,\gamma_k^2 + \tilde{C}'_\varepsilon, \label{eq:second_moment_main}
\end{eqnarray*}
with $\|{\bm \delta}_k\| \leq c_3\gamma_k$, and where $c_1 := 2b_1$, $c_3 := b\,b_2^3 N/(2b_1)$, and
\begin{equation*}
\tilde{C}'_\gamma = 4L^2b_2^4\Big(\sum_{i=1}^N \frac{\Omega_i}{\mu_i^2} + N(N\!-\!1)\Big), \quad 
\tilde{C}'_\varepsilon = b_2^2\,\bar{\sigma}_\varepsilon^2. 
\end{equation*}
\end{lemma}
\begin{proof}
See Appendix~C.
\end{proof}



\subsection{Asymptotic Convergence Rate}
Next, we present our main convergence results. 
\begin{theorem}
\label{thm:th1}
Under Assumptions~3.1--3.4 and Assumption~1, if $\sum_k \eta_k^2 < \infty$, then $\lim_{k\to\infty} \E[\|\nabla F({\bm \theta}_k)\|^2] = 0$. 
\end{theorem}
\begin{proof}
    See Appendix~D.
\end{proof}

\begin{theorem}
\label{thm:th2}
     Let $\eta_k = \eta_0 K^{-1/4}$, $\gamma_k = \gamma_0 K^{-1/4}$, where $\eta_0, \gamma_0 > 0$. Under Assumptions~3.1--3.4 and Assumption~1, after $K$ iterations:
\begin{equation*}\label{eq:rate_main}
\min_{k=1:K} \E[\|\nabla F({\bm \theta}_k)\|^2] \leq \frac{R}{\sqrt{K}} + \rho,
\end{equation*}
where $R := \frac{2\hat{\Delta}}{c_1\eta_0\gamma_0} + c_3^2\gamma_0^2 + \frac{\tilde{C}'_\gamma L\eta_0\gamma_0}{c_1}$, $\hat{\Delta} := F({\bm \theta}_0) - F({\bm \theta}^*)$ such that $F({\bm \theta}^*)=\min_{{\bm \theta}}F({\bm \theta})>-\infty$, and 
\begin{equation}\label{eq:rho_main}
\rho := \frac{\tilde{C}'_\varepsilon L\eta_0}{c_1\gamma_0} = \frac{Lb_2^2\bar{\sigma}_\varepsilon^2\eta_0}{2b_1\gamma_0}
\end{equation}
is the noise floor introduced by the encryption. Furthermore, for any $\epsilon > \rho$ and $0 < \beta < 1$, if $K = R^2/(\epsilon\beta - \rho)^2$, then $\Prob\bigl(\min_{k=1:K}\|\nabla F({\bm \theta}_k)\|^2 < \epsilon\bigr) \geq 1 - \beta$.
\end{theorem}


\begin{proof}
The global objective $F$ is $L$-smooth by Assumption~3.1 of~\cite{ezofl}, which gives the descent inequality
\begin{equation*}
F({\bm \theta}_{k+1}) \leq F({\bm \theta}_k) - \eta_k\langle\nabla F({\bm \theta}_k), {\bm g}_k\rangle + \tfrac{L}{2}\eta_k^2\|{\bm g}_k\|^2 .
\end{equation*}
Taking $\E[\cdot\mid\calH_k]$, applying Lemma~\ref{lem:lem1}, summing over $1\leq k \leq K$, and substituting $\eta_k=\eta_0K^{-1/4}$ and $\gamma_k=\gamma_0K^{-1/4}$ yields
\begin{equation*}
\min_{k=1:K}\E[\|\nabla F({\bm \theta}_k)\|^2] \leq \frac{R}{\sqrt{K}} + \rho,
\end{equation*}
where the $R/\sqrt{K}$ term originates from the initial gap, the bias, and the gradient's second moment, while the constant $\rho$ originates from the decoded HE noise $\tilde{C}'_\varepsilon$. Since $\min_{k=1:K}\|\nabla F({\bm \theta}_k)\|^2$ is nonnegative, Markov's inequality gives, for any $\epsilon>\rho$,
\begin{equation*}
\Prob\!\left(\min_{k=1:K}\|\nabla F({\bm \theta}_k)\|^2 \geq \epsilon\right) \leq \frac{1}{\epsilon}\!\left(\frac{R}{\sqrt{K}} + \rho\right),
\end{equation*}
so that, taking complements, $\Prob(\min_{k=1:K}\|\nabla F({\bm \theta}_k)\|^2 < \epsilon) \geq 1 - \frac{1}{\epsilon}(R/\sqrt{K}+\rho)$; setting $\frac{1}{\epsilon}(R/\sqrt{K}+\rho)=\beta$ gives $K = R^2/(\epsilon\beta-\rho)^2$. The detailed proof is given in Appendix~D.
\end{proof}


Theorem~\ref{thm:th2} implies that the algorithm converges to a neighborhood of size $\rho$ at rate $O(1/\sqrt{K})$, where $\rho$ is given in Eq.~\eqref{eq:rho_main}. For the specific setup used in our experiments, namely the HE parameters of Table~\ref{tab:he_params} ($\Lambda=2^{40}$, $\sigma_e=3.2$, $N=10$), identical channels with $\mu_i=\mu$ and $\Omega_i=\Omega$, and step sizes $\eta_0=\gamma_0=0.05$, the decoded noise variance is $\bar{\sigma}_\varepsilon^2 \approx 4 \times 10^{-16}\,\Omega^2$, giving a noise floor of order $\rho \approx 10^{-16}$. This is more than sixteen orders of magnitude below the channel noise variance $\sigma_n^2$, and is therefore negligible for all practical purposes. The number of iterations needed to reach $\min_{k=1:K}\|\nabla F({\bm \theta}_k)\|^2 < \epsilon$ is then $K \approx R^2/(\epsilon\beta)^2$, the same as for the non-encrypted algorithm.

\section{Numerical Results}
\label{sec:mumerical_result}

\begin{figure*}[!tbp]
\centering
\includegraphics[width=0.7\textwidth]{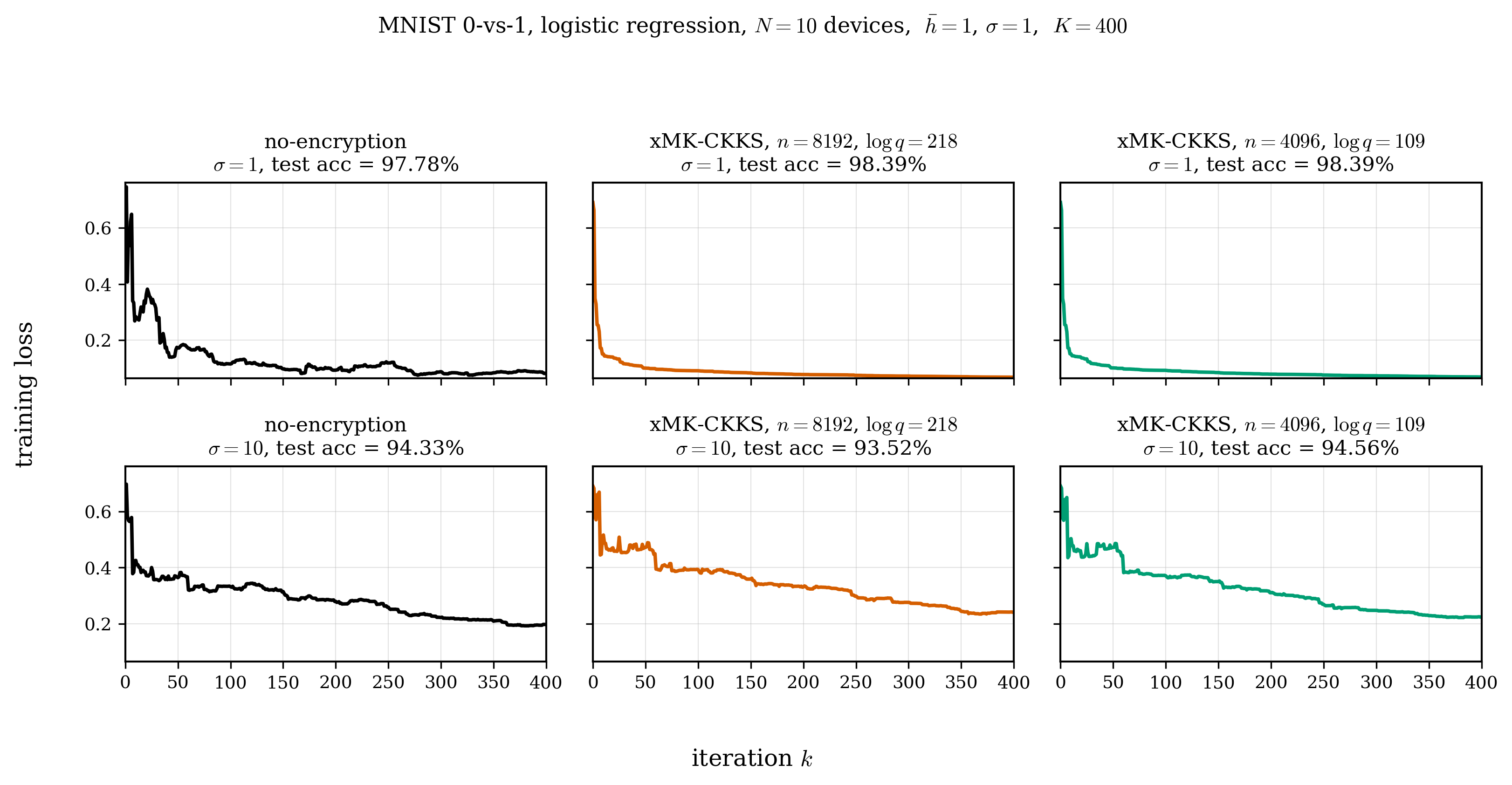}

\caption{Training loss on MNIST 0-vs-1 under the proposed protocol for $n \in \{4096, 8192\}$ and the non-encrypted baseline, in both channel regimes $\sigma \in \{1, 10\}$.}

\label{fig:setups}
\end{figure*}

\begin{figure*}[!tbp]
\centering
\includegraphics[width=0.7\textwidth]{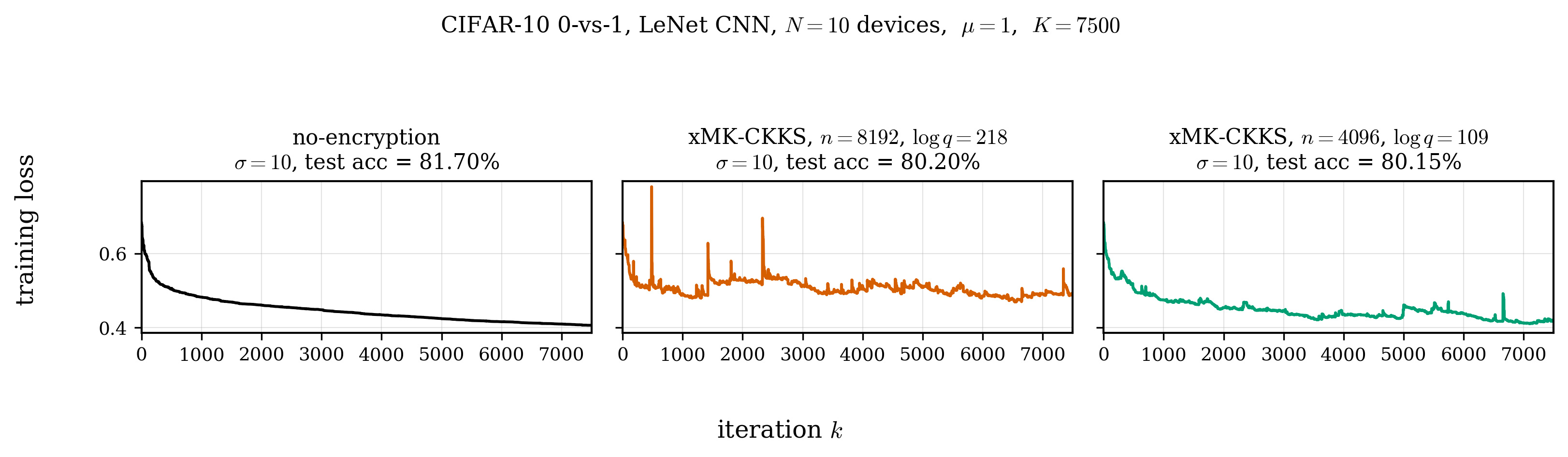}
\caption{{Training loss on CIFAR-10 0-vs-1 with a LeNet CNN ($d=61{,}241$) under the proposed protocol for $n \in \{4096, 8192\}$ and the non-encrypted baseline with $\sigma=10$}.}
\label{fig:cifar_setups}
\end{figure*}

\begin{figure}[!tbp]
\centering
\includegraphics[width=0.85\columnwidth]{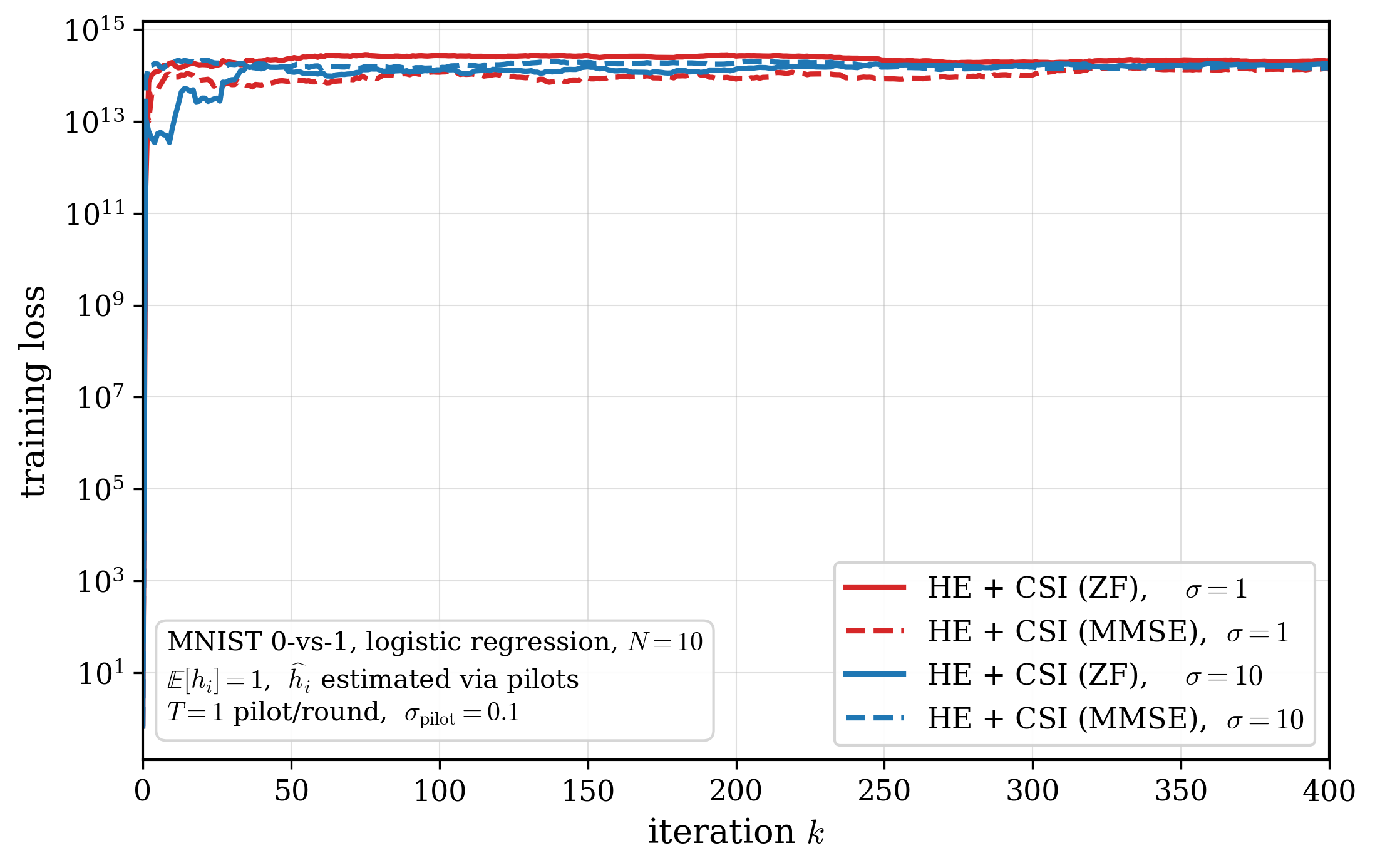}
\caption{Training loss under HE with CSI-based pre-equalization at $q \approx 2^{109}$ ($\sigma_{\mathrm{pilot}} = 0.1$). Both ZF and MMSE equalization do not converge for either channel variance, since the residual estimation error is multiplied by $q$ at decryption.}
\label{fig:eq_breaks}
\end{figure}

We evaluate the proposed protocol on MNIST 0-vs-1 binary classification with a logistic regression model ($d = 785$). The channel coefficient is $h_{i,k} \sim \mathcal{N}(\mu_i, \sigma_i^2)$, $1 \leq i \leq N$,  with $\mu_i = \mu = 1$ and $\sigma_i = \sigma \in \{1, 10\}$, giving $\Omega_i = \Omega = 1 + \sigma^2$; the channel noise standard deviation is $\sigma_n = 1$. Step sizes are $\eta_k = \eta_0(1+k)^{-0.50}$ and $\gamma_k = \gamma_0(1+k)^{-0.25}$, $1 \leq k \leq K$,  with $\eta_0 = \gamma_0 = 0.05$, batch size 128, $K = 400$ rounds, $N = 10$ devices.

Fig.~\ref{fig:setups} shows the training loss under the proposed protocol for the  two HE setups of Table~\ref{tab:he_params} and the non-encrypted baseline. Within each regime, the encrypted and non-encrypted configurations follow comparable trajectories and reach test accuracies that differ by less than $1\%$ (Table~\ref{tab:results}), confirming that the decoded HE noise  $\bar{\sigma}_\varepsilon^2 \approx 2\times 10^{-16}\,\Omega^2$  and $\bar{\sigma}_\varepsilon^2 \approx 4 \times 10^{-16}\,\Omega^2$ for the cases $n = 4096$ and $n = 8192$ is invisible to the learning algorithm.  Trajectory differences within a regime arise from independent channel realizations rather than from encryption.
We also compare to the CSI-based pre-equalization considered in Section~III-B: each device estimates $h_{i,k}$ from a pilot (with the standard deviation of the channel-estimation error $\sigma_{\mathrm{pilot}} = 0.1$) and pre-equalizes by ZF or MMSE. The residual estimation error is multiplied by $q$ at decryption, dominating the decoded gradient by a term of order $q/\Lambda \approx 2^{70}$. Fig.~\ref{fig:eq_breaks} confirms this: the training loss saturates near $10^{14}$ from the first iteration regardless of estimator or channel variance. Our proposed protocol avoids this by never estimating the channel.

\subsection*{Scaling to a Convolutional Model on CIFAR-10}
To assess the protocol on a more complex task, we repeat the 0-vs-1 binary classification experiment on CIFAR-10 with a LeNet-style CNN: two $5\times5$ convolutional layers with $6$ and $16$ output filters respectively, each followed by $2\times2$ max-pooling, and three fully connected layers ($400\!\to\!120\!\to\!84\!\to\!1$) with ReLU activations, trained on a single output logit with a binary cross-entropy loss. The model has $d = 61{,}241$ trainable parameters. The hidden layers use Kaiming initialization matched to the ReLU nonlinearities, and the output layer is initialized with small weights so that the initial logits remain close to zero. The channel and HE parameters are unchanged ($N=10$, $\mu_i = 1$, $\sigma=10$, $\sigma_n = 1$, in addition to the two parameter sets of Table~\ref{tab:he_params}); the step sizes are $\eta_k = 0.005(1+k)^{-0.50}$ and $\gamma_k = 0.01(1+k)^{-0.25}$, and training runs for $K = 7500$ rounds.

The smaller step sizes are dictated by the model dimension, not by the encryption. With Rademacher perturbation entries, $b_2 = \|{\bm \Phi}_k\| = \sqrt{d}$ grows from $28.0$ on MNIST to $247.5$ on CIFAR-10, and since by  Lemma~\ref{lem:lem1} $\tilde{C}'_\gamma \propto b_2^4 = d^2$, the dimension-dependent factor of the second-moment bound grows by $(61{,}241/785)^2 \approx 6.1\times10^{3}$; the initial step sizes are reduced accordingly. This dependence comes from the ZO estimator itself and affects the encrypted and non-encrypted variants identically.

{Fig.~\ref{fig:cifar_setups} shows the training loss for the two HE setups and the non-encrypted baseline. The test accuracy drops from the $93$--$94\%$ range to the $80$--$81\%$ range (Table~\ref{tab:results}). Nevertheless, within each channel regime the encrypted and non-encrypted configurations follow comparable trajectories and reach test accuracies that differ by at most $1.6\%$; as in Fig.~\ref{fig:setups}, trajectory differences arise from independent channel realizations rather than from encryption. Most importantly, the encrypted payload remains a single scalar per device per round: the ciphertext size, uplink load, and encryption time of Table~\ref{tab:he_params} are unchanged, confirming that the overhead of the proposed protocol is independent of the model dimension.}

\begin{table}[!htb]
\centering
\caption{{Test accuracy across the HE parameter setups and the non-encrypted baseline, on MNIST 0-vs-1 (logistic regression, $d=785$) and CIFAR-10 0-vs-1 (LeNet CNN, $d=61{,}241$).}}
\label{tab:results}
\footnotesize
{
\begin{tabular}{@{}lcccccc@{}}
\toprule
& & & \multicolumn{2}{c}{\textbf{MNIST}} & \multicolumn{2}{c@{}}{\textbf{CIFAR-10}} \\
\cmidrule(lr){4-5}\cmidrule(l){6-7}
Setup & $n$ & $\log_2 q$ & $\sigma\!=\!1$ & $\sigma\!=\!10$ & $\,\,\,\,\sigma\!=\!10$ \\
\midrule
A (HE) & 8192 & 218 & 98.39\% & 93.52\%  & \,\,\,\,80.20\% \\
B (HE) & 4096 & 109 & 98.30\% & 94.56\%  & \,\,\,\,80.15\% \\
no-enc & --- & --- & 97.78\% & 94.33\%  & \,\,\,\,81.70\% \\
\bottomrule
\end{tabular}
}
\end{table}

\section{Conclusion}

We showed that multi-key HE and over-the-air aggregation can coexist: the additive structure of xMK-CKKS matches the wireless superposition, and the encryption noise vanishes after decoding. Among $N$ users, each user's device incorporates the known channel mean $\{\mu_i\}_{1\leq i\leq N}$ into its encoding, and the gradient estimate is formed directly from the decrypted aggregate. The decoded HE noise introduces a noise floor 
that is negligible for all practical purposes. The conducted numerical experiments on MNIST and CIFAR-10 confirm this conclusion, and the ZO structure keeps the encrypted payload at a single scalar per device, so the encryption overhead does not grow with the model dimension~$d$.

Several directions remain open: reducing the key retransmission overhead by exploiting slow channel variation, extending the protocol to approximate block fading with controlled error, and combining HE with differential privacy where the channel noise serves as a privacy mechanism on top of the encryption layer.

\bibliographystyle{IEEEtran}
\bibliography{references}

\appendices

\section{Detailed Correctness Verification of xMK-CKKS}
\label{appendix:correctness_xmk_ckks}

This appendix provides a detailed derivation of the correctness of the collaborative decryption procedure in the xMK-CKKS scheme and analyzes the resulting bounded noise terms.

\subsection{Aggregated Public Key}

Each device $d_i$ generates a secret key $s_i$ and computes its partial public key:
\[
b_i=-s_i a+e_i \mod q,
\]
where:
\begin{itemize}
\item $a$ is a common public polynomial,
\item $e_i$ is a small RLWE error sampled from a Gaussian distribution.
\end{itemize}

The aggregated public key is:
\[
\tilde{b}=\sum_{i=1}^{N}b_i
=
-\sum_{i=1}^{N}s_i a
+
\sum_{i=1}^{N}e_i
\mod q.
\]

Thus:
\[
\tilde{b}+a\sum_{i=1}^{N}s_i
=
\sum_{i=1}^{N}e_i.
\]

\subsection{Encryption}

Each device encrypts plaintext $m_i$ as:
\[
ct_i=(c_{0,i},c_{1,i}),
\]
where:
\[
c_{0,i}=v_i\tilde{b}+m_i+e_i^{(0)},
\]
\[
c_{1,i}=v_i a+e_i^{(1)}.
\]

Here:
\begin{itemize}
\item $v_i$ is a random masking polynomial,
\item $e_i^{(0)}$ and $e_i^{(1)}$ are bounded Gaussian noise terms.
\end{itemize}

\subsection{Homomorphic Aggregation}

The server computes:
\[
C_{\mathrm{sum}}
=
\sum_{i=1}^{N}ct_i
=
(C_{\mathrm{sum},0},C_{\mathrm{sum},1}),
\]
with:
\[
C_{\mathrm{sum},0}
=
\sum_{i=1}^{N}
(v_i\tilde{b}+m_i+e_i^{(0)}),
\]
\[
C_{\mathrm{sum},1}
=
\sum_{i=1}^{N}
(v_i a+e_i^{(1)}).
\]

\subsection{Collaborative Decryption Shares}

Each device computes a partial decryption share:
\[
D_i=s_i C_{\mathrm{sum},1}+e_i^{*},
\]
where $e_i^{*}$ is an additional bounded masking noise.

Substituting $C_{\mathrm{sum},1}$:
\[
D_i
=
s_i
\sum_{j=1}^{N}
(v_j a+e_j^{(1)})
+
e_i^{*}.
\]

\subsection{Server-Side Reconstruction}

The server reconstructs:
\begin{align*}
&C_{\mathrm{sum},0}
+
\sum_{i=1}^{N}D_i\\
&=
\sum_{i=1}^{N}
(v_i\tilde{b}+m_i+e_i^{(0)})
+
\sum_{i=1}^{N}
s_i
\sum_{j=1}^{N}
(v_j a+e_j^{(1)})
+
\sum_{i=1}^{N}e_i^{*},
\end{align*}
where we substituted all terms to get the equality.

\subsection{Expansion of the Aggregated Public Key}

Using:
\[
\tilde{b}
=
-\sum_{j=1}^{N}s_j a
+
\sum_{j=1}^{N}e_j,
\]
we obtain:
\[
\sum_{i=1}^{N}v_i\tilde{b}
=
-
\sum_{i=1}^{N}\sum_{j=1}^{N}v_i s_j a
+
\sum_{i=1}^{N}\sum_{j=1}^{N}v_i e_j.
\]

Thus:
\[
C_{\mathrm{sum},0}
+
\sum_{i=1}^{N}D_i
=
-
\sum_{i=1}^{N}\sum_{j=1}^{N}v_i s_j a
+
\sum_{i=1}^{N}\sum_{j=1}^{N}v_i e_j
\]
\[
+
\sum_{i=1}^{N}m_i
+
\sum_{i=1}^{N}e_i^{(0)}
+
\sum_{i=1}^{N}\sum_{j=1}^{N}s_i(v_j a+e_j^{(1)})
+
\sum_{i=1}^{N}e_i^{*}.
\]

\subsection{Cancellation of Key-Dependent Terms}

Rearranging:
\[
=
-
\sum_{i=1}^{N}\sum_{j=1}^{N}v_i s_j a
+
\sum_{i=1}^{N}\sum_{j=1}^{N}s_i v_j a
\]
\[
+
\sum_{i=1}^{N}m_i
+
\sum_{i=1}^{N}\sum_{j=1}^{N}v_i e_j
+
\sum_{i=1}^{N}e_i^{(0)}
\]
\[
+
\sum_{i=1}^{N}\sum_{j=1}^{N}s_i e_j^{(1)}
+
\sum_{i=1}^{N}e_i^{*}.
\]

The large key-dependent masking terms cancel algebraically:
\[
-
\sum_{i=1}^{N}\sum_{j=1}^{N}v_i s_j a
+
\sum_{i=1}^{N}\sum_{j=1}^{N}s_i v_j a
=
0.
\]

\subsection{Bounded Noise Analysis}

The remaining expression becomes:
\[
=
\sum_{i=1}^{N}m_i
+
\underbrace{
\sum_{i=1}^{N}\sum_{j=1}^{N}v_i e_j
+
\sum_{i=1}^{N}e_i^{(0)}
+
\sum_{i=1}^{N}\sum_{j=1}^{N}s_i e_j^{(1)}
+
\sum_{i=1}^{N}e_i^{*}
}_{\text{bounded RLWE noise}}.
\]

All remaining noise terms are bounded because:
\begin{itemize}
\item the secret keys $s_i$ are sampled from small distributions,
\item the masking polynomials $v_i$ are bounded,
\item the error terms $e_i$, $e_i^{(0)}$, $e_i^{(1)}$, and $e_i^{*}$ follow bounded discrete Gaussian distributions.
\end{itemize}

Therefore, as long as the accumulated noise magnitude remains below the CKKS decoding threshold, correct decryption is preserved.

\subsection{Correctness Result}

The final reconstructed value satisfies:
\[
C_{\mathrm{sum},0}
+
\sum_{i=1}^{N}D_i
\mod q
=
\sum_{i=1}^{N}m_i
+
\text{bounded noise}.
\]

Hence:
\[
C_{\mathrm{sum},0}
+
\sum_{i=1}^{N}D_i
\approx
\sum_{i=1}^{N}m_i.
\]

This establishes the correctness of collaborative decryption in xMK-CKKS.

\section{Noise Budget}
\label{app:noise}


Let $p, r \in \Rq$ be independent zero-mean polynomials with i.i.d.\ coefficients of variances $\sigma_p^2$, $\sigma_r^2$. Then

\begin{equation}
\label{eq:rpr}
\E\|pr\|^2 = n^2\sigma_p^2\sigma_r^2.
\end{equation}

We refer to Eq~\eqref{eq:rpr} as the ring product rule. 
\subsection{First Moment: \texorpdfstring{$\E[\hat{\varepsilon}_k] = 0$}{E[epsilon] = 0}}

The seven terms of $\hat{\varepsilon}_k$ from the recovery Eq~\eqref{eq:Mhat_expand} are
\begin{align}
\hat{\varepsilon}_k &= \underbrace{V_k\tilde{E}_{\mathrm{agg},k}}_{(A)} + \underbrace{\textstyle\sum_i h_{i,k}e_{0,i,k}}_{(B1)} + \underbrace{w_{0,k}}_{(B2)} + \underbrace{\tilde{S}_k\textstyle\sum_i h_{i,k}e_{1,i,k}}_{(C1)} \notag\\
&\quad + \underbrace{\tilde{S}_k w_{1,k}}_{(C2)} + \underbrace{\textstyle\sum_j h_{j,k}e_{j,k}^*}_{(E)} + \underbrace{w_{D,k}}_{(D)},
\label{eq:seven}
\end{align}
where $V_k = \sum_i h_{i,k}v_{i,k}$, $\tilde{S}_k = \sum_j h_{j,k}s_j$, $\tilde{E}_{\mathrm{agg},k} = \sum_j h_{j,k}e_j + w_{b,k}$.

Conditioning on $h_{\cdot,k}$ and using Assumption~1, each term has zero conditional mean:
\begin{align*}
\E[(A)\mid h_{\cdot,k}] &= \Big(\textstyle\sum_i h_{i,k}\E[v_{i,k}]\Big)\,\E[\tilde{E}_{\mathrm{agg},k}] = 0, \\
\E[(B1)\mid h_{\cdot,k}] &= \textstyle\sum_i h_{i,k}\,\E[e_{0,i,k}] = 0, \quad \E[(B2)] = 0,\\
\E[(C1)\mid h_{\cdot,k}] &= \tilde{S}_k\textstyle\sum_i h_{i,k}\,\E[e_{1,i,k}] = 0, \\
\E[(C2)\mid h_{\cdot,k}] &= \tilde{S}_k\E[w_{1,k}] = 0, \quad \E[(D)] = \E[w_{D,k}] = 0\\
\E[(E)\mid h_{\cdot,k}] &= \textstyle\sum_i h_{i,k}\,\E[e_{i,k}^*] = 0.
\end{align*}
Averaging over $h_{\cdot,k}$, we get $\E[\hat{\varepsilon}_k]=0$. As $\hat{\varepsilon}_k$ is independent of the model history, the same holds conditionally, thus yielding Proposition~\ref{prop:prop1}.

\subsection{Second Moment: Term-by-Term}
Each of the seven terms in Eq~\eqref{eq:seven} is a product of independent, zero-mean polynomials, so its second moment is computed in two stages: first, conditioning on $h_{\cdot,k}$, the ring product rule Eq~\eqref{eq:rpr} is applied to the HE and channel noises; then the expectation over $h_{\cdot,k}$ is applied. For a Gaussian channel, $\E[h_{i,k}^4]=3\Omega_i^2-2\mu_i^4$, and by independence across clients
\begin{equation}\label{eq:ch5}
\E\!\left[\!\Big(\textstyle\sum_i h_{i,k}^2\Big)^{\!2}\right] = \Big(\textstyle\sum_i\Omega_i\Big)^2 + 2\textstyle\sum_i\Omega_i^2 - 2\textstyle\sum_i\mu_i^4.
\end{equation}

Denote $S_\Omega:=\sum_i\Omega_i$, $S_{\Omega^2}:=\sum_i\Omega_i^2$, $S_{\mu^4}:=\sum_i\mu_i^4$, so $\E[(\sum_i h_{i,k}^2)^2]=S_\Omega^2+2S_{\Omega^2}-2S_{\mu^4}$ and $\E[\sum_i h_{i,k}^2]=S_\Omega$. The seven terms then evaluate to
\begin{align}
\E\|(A)\|^2 &= \tfrac{2}{3}n^2\!\left[\sigma_e^2\bigl(S_\Omega^2+2S_{\Omega^2}-2S_{\mu^4}\bigr)+\sigma_w^2 S_\Omega\right], \notag\\
\E\|(B1)\|^2 &= n\sigma_e^2 S_\Omega, \qquad \E\|(B2)\|^2 = n\sigma_w^2, \nonumber\\
\E\|(C1)\|^2 &= n^2\sigma_e^2\bigl(S_\Omega^2+2S_{\Omega^2}-2S_{\mu^4}\bigr), \notag\\
\E\|(C2)\|^2 &= n^2\sigma_w^2 S_\Omega, \qquad \E\|(D)\|^2 = n\sigma_w^2, \nonumber\\
\E\|(E)\|^2 &= n\sigma_\phi^2 S_\Omega. \label{eq:C1_final}
\end{align}
\subsection{Summing and Decoding}
Using the triangular inequality $\E\|\sum_{\ell=1}^7 X_\ell\|^2 \leq \sum_{\ell=1}^7 \E\|X_\ell\|^2$ and summing the identities in Eq~\eqref{eq:C1_final} gives
\begin{align}
B_\varepsilon^2 &= \tfrac{2}{3}n^2\bigl(\sigma_e^2(S_\Omega^2+2S_{\Omega^2}-2S_{\mu^4})+\sigma_w^2 S_\Omega\bigr) \notag\\
&\;+ n\sigma_e^2 S_\Omega + n\sigma_w^2 + n^2\sigma_e^2(S_\Omega^2+2S_{\Omega^2}-2S_{\mu^4}) \notag\\
&\;+ n^2\sigma_w^2 S_\Omega + n\sigma_\phi^2 S_\Omega + n\sigma_w^2. \label{eq:Beps}
\end{align}
The per-coefficient decoded variance is $\bar{\sigma}_\varepsilon^2 := B_\varepsilon^2/(n\Lambda^2)$. In the homogeneous case $\mu_i=\mu$, $\Omega_i=\Omega$ ($S_\Omega=N\Omega$, $S_{\Omega^2}=N\Omega^2$, $S_{\mu^4}=N\mu^4$), this gives $\bar{\sigma}_\varepsilon^2 \approx 6\times10^{-17}\,\Omega^2$ and $\bar{\sigma}_\varepsilon^2 \approx 3\times10^{-17}\,\Omega^2$ for the parameters of \ref{tab:he_params}.

For decryption to be correct, the CKKS decryption mechanism requires $\|\hat{\varepsilon}_k\|_\infty < q/2$; otherwise, the modular reduction wraps around and the message cannot be recovered \cite{mkckks}. By $\|\hat{\varepsilon}_k\|_\infty \leq \|\hat{\varepsilon}_k\|$ and Markov's inequality on $\|\hat{\varepsilon}_k\|^2$:
$\Prob(\|\hat{\varepsilon}_k\|_\infty \geq q/2) \leq 4B_\varepsilon^2/q^2.$
With $B_\varepsilon^2$ from Eq~\eqref{eq:Beps}, which is independent of $q$, and $q \approx 2^{109}, 2^{218}$, this probability is negligible.

\section{Proof of Lemma~\ref{lem:lem1}}\label{app:lem1}
\label{app:bias}

The constants $L$, $b$, $b_1$, $b_2$ are given by Assumptions~3.1--3.4 of~\cite{ezofl}: $L$ is the Lipschitz constant of the gradient $\nabla F_i$, i.e.\ $\|\nabla F_i({\bm x})-\nabla F_i({\bm y})\|\leq L\|{\bm x}-{\bm y}\|$; $b$ is a uniform bound on the Hessian  ($\|\nabla^2 F_i\|_2 \leq b$); $b_1 = \E[({ \Phi}_k^j)^2]$, which is a single constant because ${\bm \Phi}_k$ has i.i.d.\ entries drawn from the same fixed distribution in every round, so all coordinates $j$ share the same second moment and it does not vary with $k$; and $b_2$ is a uniform upperbound on $\|{\bm \Phi}_k\| \leq b_2$. The derived constants are $c_1 = 2b_1$ and $c_3 = bb_2^3 N/(2b_1)$. 
Recall that the history of models and samples up to and including ${\bm \theta}_k$ is denoted by
\begin{equation}\label{eq:Hk_def}
\calH_k := \{{\bm \theta}_0, \xi_0, \ldots, {\bm \theta}_{k-1}, \xi_{k-1}, {\bm \theta}_k\},
\end{equation}
where $\xi_k := \{\xi_{i,k}\}_{i=1}^N$ collects the round-$k$ samples across devices.
We note that by Assumption~3.4 of~\cite{ezofl} and Assumption~1, the variables $h_{\cdot,k}$, ${\bm \Phi}_k$, $\xi_{\cdot,k}$,  and the HE randomness are independent of $\calH_k$, hence conditioning on $\calH_k$ preserves mutual independence.


We proceed with the proof of Lemma~\ref{lem:lem1}. The gradient estimator is given by 
\begin{equation}\label{eq:gk_expand}
{\bm g}_k = {\bm \Phi}_k\left(\sum_i \frac{h_{i,k}}{\mu_i}\Delta f_{i,k} + \varepsilon_k\right).
\end{equation}

\paragraph*{Bias}
Since $h_{i,k}$ is independent of $({\bm \Phi}_k, \xi_{\cdot,k})$ and $\varepsilon_k$ is zero-mean and independent of $({\bm \Phi}_k, \xi_{\cdot,k}, h_{\cdot,k})$ (Assumption~1), and since the independence structure is unchanged given $\calH_k$, as noted after Eq~\eqref{eq:Hk_def}, taking $\E[\cdot|\calH_k]$ gives
\begin{align}
\E[{\bm g}_k|\calH_k] &= \sum_{i=1}^N \frac{\E[h_{i,k}]}{\mu_i}\,\E_{\Phi,\xi}[{\bm \Phi}_k\Delta f_{i,k}|\calH_k] \notag\\
&\quad + \underbrace{\E[\varepsilon_k]}_{=\,0}\,\E[{\bm \Phi}_k|\calH_k] \notag\\
&= \sum_{i=1}^N \E_{\Phi,\xi}[{\bm \Phi}_k\Delta f_{i,k}|\calH_k], \label{eq:bias_step1}
\end{align}

where the last equality uses $\E[h_{i,k}] = \mu_i$ from~\eqref{eq:channel_stats}.

We follow the same two-point ZO argument as in~\cite[Appendix~A-A]{ezofl} (a mean-value Taylor expansion of $F_i({\bm \theta}_k \pm \gamma_k{\bm \Phi}_k)$, we obtain  

\begin{equation}\label{eq:first_moment_bound}
\|\E[{\bm g}_k|\calH_k]\| \leq c_1\gamma_k(\|{\nabla} F({\bm \theta}_k)\| + c_3\gamma_k),
\end{equation}
which completes the proof of the first part of Lemma~\ref{lem:lem1}. 

We proceed with the proof of the second moment result. Since $\|{\bm \Phi}_k\| \leq b_2$, Eq~\eqref{eq:gk_expand} implies
\begin{equation}\label{eq:sm_start}
\E[\|{\bm g}_k\|^2|\calH_k] \leq b_2^2\,\E\!\left[\left(\sum_i \frac{h_{i,k}}{\mu_i}\Delta f_{i,k} + \varepsilon_k\right)^{\!2}\Big|\calH_k\right].
\end{equation}

Expanding the square and noticing that $\E[\varepsilon_k] = 0$ and that $\varepsilon_k$ is independent of $(h_{\cdot,k}, {\bm \Phi}_k, \xi_{\cdot,k})$ by Assumption~1, the cross term vanishes and we get 
\begin{align}
&\E\!\left[\left(\sum_i \frac{h_{i,k}}{\mu_i}\Delta f_{i,k} + \varepsilon_k\right)^{\!2}\Big|\calH_k\right] \notag\\
&= \E\!\left[\left(\sum_i \frac{h_{i,k}}{\mu_i}\Delta f_{i,k}\right)^{\!2}\Big|\calH_k\right] + \bar{\sigma}_\varepsilon^2. \label{eq:square_expand}
\end{align}
For the first term, we expand the double sum and use the fact that $h_{i,k}$ is independent of $({\bm \Phi}_k, \xi_{\cdot,k})$, and $h_{i,k}$, $h_{j,k}$ are independent for $i \neq j$. We obtain
\begin{align}
&\E\!\left[\left(\sum_i \frac{h_{i,k}}{\mu_i}\Delta f_{i,k}\right)^{\!2}\Big|\calH_k\right]\nonumber\\
&= \sum_i\sum_j \frac{\E[h_{i,k}h_{j,k}]}{\mu_i\mu_j}\,\E_{\Phi,\xi}[\Delta f_{i,k}\Delta f_{j,k}|\calH_k] \notag\\
&= \sum_i \frac{\Omega_i}{\mu_i^2}\,\E[(\Delta f_{i,k})^2|\calH_k] + \sum_{i \neq j} \frac{\mu_i\mu_j}{\mu_i\mu_j}\,\E[\Delta f_{i,k}\Delta f_{j,k}|\calH_k], \label{eq:main_expand}
\end{align}
where we used $\E[h_{i,k}^2] = \Omega_i$ and $\E[h_{i,k}h_{j,k}] = \mu_i\mu_j$ for $i \neq j$.

For the diagonal terms, by the Lipschitz bound $|\Delta f_{i,k}| \leq 2Lb_2\gamma_k$ (Assumption~3.2):
\begin{equation}\label{eq:lip1}
\E[(\Delta f_{i,k})^2|\calH_k] \leq 4L^2b_2^2\gamma_k^2.
\end{equation}

For the cross terms ($i \neq j$), since $\xi_{i,k}$ and $\xi_{j,k}$ are independent conditioned on ${\bm \Phi}_k$:
\begin{align}
|\E[\Delta f_{i,k}\Delta f_{j,k}|\calH_k]| &\leq \E|[\Delta f_{i,k}||\Delta f_{j,k}||\calH_k] \leq 4L^2b_2^2\gamma_k^2, \label{eq:lip2}
\end{align}
where the first inequality is due to Jensen $|\E[\cdot]|\le\E|\cdot|$, 
and where we used $|\Delta f_{i,k}|\le 2Lb_2\gamma_k$ (Assumption~3.2) to write the second inequality.

Substituting Eq~\eqref{eq:lip1} and Eq~\eqref{eq:lip2} into Eq~\eqref{eq:main_expand} implies
\begin{align}
&\E\!\left[\left(\sum_i \frac{h_{i,k}}{\mu_i}\Delta f_{i,k}\right)^{\!2}\Big|\calH_k\right] \nonumber\\
&\leq \sum_i \frac{\Omega_i}{\mu_i^2}\, 4L^2b_2^2\gamma_k^2 + N(N\!-\!1)\, 4L^2b_2^2\gamma_k^2 \notag\\
&= 4L^2b_2^2\gamma_k^2\Big(\sum_i \frac{\Omega_i}{\mu_i^2} + N(N\!-\!1)\Big). \label{eq:main_bound}
\end{align}

Using Eq.~\eqref{eq:main_bound} 
in Eq.~\eqref{eq:sm_start}, we get
\begin{equation}\label{eq:sm_final}
\E[\|{\bm g}_k\|^2|\calH_k] \leq \tilde{C}'_\gamma\,\gamma_k^2 + \tilde{C}'_\varepsilon,
\end{equation}
where
\begin{align}
\tilde{C}'_\gamma &:= 4L^2b_2^4\Big(\sum_i \frac{\Omega_i}{\mu_i^2} + N(N\!-\!1)\Big), \label{eq:Cgamma}\\
\tilde{C}'_\varepsilon &:= b_2^2\,\bar{\sigma}_\varepsilon^2, \label{eq:Ceps}
\end{align}
where $\bar{\sigma}_\varepsilon^2 := B_\varepsilon^2/(n\Lambda^2)$ is defined in Appendix~\ref{app:noise}-C. In the identical channles case $\mu_i=\mu$, $\Omega_i=\Omega$, this reduces to $\tilde{C}'_\gamma = 4NL^2b_2^4(\Omega+(N-1)\mu^2)/\mu^2$.
\section{Proofs of Theorems~\ref{thm:th1} and~\ref{thm:th2}}
\label{app:rate}

By $L$-smoothness of $F$ (Assumption~3.1 of~\cite{ezofl}):

\begin{equation}\label{eq:smoothness}
F({\bm \theta}_{k+1}) \leq F({\bm \theta}_k) - \eta_k\langle \nabla F({\bm \theta}_k), {\bm g}_k\rangle + \frac{L}{2}\eta_k^2\|{\bm g}_k\|^2.
\end{equation}

Taking $\E[\cdot|\calH_k]$ and applying Lemma~\ref{lem:lem1}, we get
\begin{align}
&\E[F({\bm \theta}_{k+1})|\calH_k] \notag\\
&\leq F({\bm \theta}_k) - c_1\eta_k\gamma_k\langle{\nabla} F({\bm \theta}_k), {\nabla} F({\bm \theta}_k) + {\bm \delta}_k\rangle \notag\\
&\quad + \frac{L}{2}\eta_k^2(\tilde{C}'_\gamma\gamma_k^2 + \tilde{C}'_\varepsilon) \notag\\
&\leq F({\bm \theta}_k) - \frac{c_1\eta_k\gamma_k}{2}\|{\nabla} F({\bm \theta}_k)\|^2 + \frac{c_1c_3^2}{2}\eta_k\gamma_k^3 \notag\\
&\quad + \frac{\tilde{C}'_\gamma L}{2}\eta_k^2\gamma_k^2 + \frac{\tilde{C}'_\varepsilon L}{2}\eta_k^2, \label{eq:descent}
\end{align}
where we used $-\langle a,b\rangle \leq \frac{1}{2}\|a\|^2 + \frac{1}{2}\|b\|^2$ and $\|{\bm \delta}_k\| \leq c_3\gamma_k$ in order to write Eq~\eqref{eq:descent}.

Summing from $k = 0$ to $K$, taking full expectation over all possible history $\calH_K$ and using the tower property as done in~\cite{ezofl}, we obtain
\begin{align}
&\frac{c_1}{2}\sum_{k=0}^K\eta_k\gamma_k\E[\|{\nabla} F({\bm \theta}_k)\|^2] \leq \hat{\Delta} + \frac{c_1c_3^2}{2}\sum_{k=0}^K\eta_k\gamma_k^3 \notag\\
&\quad + \frac{\tilde{C}'_\gamma L}{2}\sum_{k=0}^K\eta_k^2\gamma_k^2 + \frac{\tilde{C}'_\varepsilon L}{2}\sum_{k=0}^K\eta_k^2, \label{eq:tele}
\end{align}
where $\hat{\Delta} = F({\bm \theta}_0) - F({\bm \theta}^*)$.

\textit{Proof of Theorem~1.} By Assumption~3.3 of \cite{ezofl}, $\sum_k\eta_k\gamma_k^3 < \infty$ and $\sum_k\eta_k^2\gamma_k^2 < \infty$. The HE noise term requires $\sum_k\eta_k^2 < \infty$, which holds for $\eta_k = \eta_0(1+k)^{-\upsilon_1}$ with $\upsilon_1 > 1/2$. The RHS of Eq~\eqref{eq:tele} is then finite. Since $\sum_k\eta_k\gamma_k$ diverges by Assumption~3.3, we can show in a similar way to~\cite[Appendix~A-C]{ezofl} that $\lim_{k\to\infty}\E[\|{\nabla} F({\bm \theta}_k)\|^2] = 0$.

\textit{Proof of Theorem~2.} Set $\eta_k = \eta = \eta_0 K^{-1/4}$, $\gamma_k = \gamma = \gamma_0 K^{-1/4}$. Using Eq~\eqref{eq:tele}
\begin{align}
&K\eta_0\gamma_0 K^{-1/2}\min_{k=1:K}\E[\|{\nabla} F({\bm \theta}_k)\|^2] \notag\\
&\leq \frac{2\hat{\Delta}}{c_1} + c_3^2\eta_0\gamma_0^3 + \frac{\tilde{C}'_\gamma L}{c_1}\eta_0^2\gamma_0^2 + \frac{\tilde{C}'_\varepsilon L}{c_1}\eta_0^2 K^{1/2}. \label{eq:rate_sub}
\end{align}

Dividing by $\eta_0\gamma_0 K^{1/2}$
\begin{align}
&\min_{k=1:K}\E[\|{\nabla} F({\bm \theta}_k)\|^2] \notag\\
&\leq \underbrace{\frac{2\hat{\Delta}}{c_1\eta_0\gamma_0\sqrt{K}} + \frac{c_3^2\gamma_0^2}{\sqrt{K}} + \frac{\tilde{C}'_\gamma L\eta_0\gamma_0}{c_1\sqrt{K}}}_{= O(1/\sqrt{K})} + \underbrace{\frac{\tilde{C}'_\varepsilon L\eta_0}{c_1\gamma_0}}_{\rho}. \label{eq:rate}
\end{align}

By Markov's inequality, for any $\epsilon > \rho$ and $\beta > 0$:
\begin{align}
&\Prob\!\left(\min_{k=1:K}\|{\nabla} F({\bm \theta}_k)\|^2 \geq \epsilon\right) \leq \frac{1}{\epsilon}\left(\frac{R}{\sqrt{K}} + \rho\right), \label{eq:markov}
\end{align}
where $R = \frac{2\hat{\Delta}}{c_1\eta_0\gamma_0} + c_3^2\gamma_0^2 + \frac{\tilde{C}'_\gamma L\eta_0\gamma_0}{c_1}$. Setting the RHS of~\eqref{eq:markov} equal to $\beta$ and solving, we get
\begin{equation*}\label{eq:K_formula}
K = \frac{R^2}{(\epsilon\beta - \rho)^2},
\end{equation*}
which completes the proof of Theorem~\ref{thm:th2}.

\end{document}